\documentclass[onecolumn,draftclsnofoot]{IEEEtran}
\usepackage{mathtools, cuted}
\usepackage{graphicx}
\usepackage{mathtools}
\usepackage{amsmath,amssymb, amsthm, bbm}
\usepackage{epsfig}
\usepackage{varioref}
\usepackage{subfigure}
\usepackage[sort&compress,square,numbers]{natbib}
\usepackage{arydshln}
\usepackage{paralist}
\usepackage[colorlinks]{hyperref}
\usepackage{algpseudocode}
\usepackage{algorithm}
\usepackage{color}
\usepackage{setspace}

\newcommand{\xvec}{{\bf{x}}}

\newcommand{\svec}{{\bf{s}}}

\newcommand{\Dmat}{{\bf{D}}}

\newcommand{\define}{\stackrel{\triangle}{=}}

\def\thetavec{{\mbox{\boldmath $\theta$}}}

\def\thetavecsmall{{\mbox{\boldmath {\scriptsize $\theta$}}}}

\newcommand{\be}{\begin{equation}}
\newcommand{\ee}{\end{equation}}
\newcommand{\beqna}{\begin{eqnarray}}
\newcommand{\eeqna}{\end{eqnarray}}

\newtheorem{theorem}{Theorem}

\newtheorem{definition}{Definition}

\newcommand{\ud}{\,\mathrm{d}}

\begin{document}
\title{Low-Complexity Detection of Small Frequency Changes
by the Generalized LMPU Test}
\author{Eyal Levy and Tirza Routtenberg$^*$,\\
School of Electrical and
Computer Engineering,\\
Ben-Gurion University of the Negev, Beer-Sheva
84105, Israel
	\thanks{ *Corresponding author. E-mail address: tirzar@bgu.ac.il.}
}

\maketitle
\vspace{-1.5cm}
\begin{abstract}
In this paper, we consider the detection of a small change in the frequency
of  sinusoidal signals, which arises in various signal processing applications.
The  generalized likelihood ratio test (GLRT) for this problem
uses the maximum likelihood (ML) estimator of the frequency, and therefore suffers from high computational complexity.
In addition, the GLRT is not necessarily optimal and its performance may degrade for non-asymptotic scenarios that are characterized by close hypotheses and small sample sizes.
 	In this paper we propose a new detection method,
 named the generalized locally most powerful unbiased (GLMPU) test, which is
  a general method for local detection  in the presence of nuisance parameters.
  A closed-form expression of the GLMPU test is developed for the detection of frequency deviation
  in the case where the complex amplitudes of the measured signals are unknown.
Numerical simulations show improved performance over
the GLRT in terms of probability of detection performance and computational complexity. 
\end{abstract}
\vspace{-0.75cm}
\vspace{-0.75cm}
\section{Introduction}
\label{sec:intro}
The problem of  detection of small frequency deviations in sinusoidal signals with unknown complex amplitudes arises in many applications such as sonar, communications,  and power systems \cite{robby}. For example, in power systems small frequency changes can be a precursor to various faults and contingencies  \cite{attack,smartGrid,smartGrid2}.
Another example is in atomic clocks, which are  widely employed in current electronic systems for guaranteeing accurate synchronization and/or high stability of the time reference. One of the most important types of faults that may affect the atomic clock behavior is the frequency jump \cite{barto,atomc}.  
 Thus, the ability to detect and track varying frequency deviation is highly desirable for attaining robustness and sustainability in various systems.

The detection of frequency deviation with unknown amplitudes is a special case of 
composite hypothesis testing, in which the likelihood depends on unknown parameters.
For the general case of composite hypothesis testing,  the  uniformly most powerful (UMP) test does not usually exist \cite{Ali}.
Instead,
the generalized likelihood ratio test (GLRT) is  widely used  due to its ease of implementation and its asymptotic properties
  \cite{Lehmann,Kay_detection}.
However, the GLRT uses the maximum likelihood (ML) estimator, and therefore suffers from high complexity in nonlinear models.
Moreover, the GLRT is not optimal in the Neyman–Pearson sense \cite{6376186} and its  performance may degrade for non-asymptotic scenarios that are characterized by close hypotheses, small sample sizes, or mismatched models {\cite{Kim_Hero_2000,6376186,Friedmann_Fishler_Messer_2002,Liu_Nehorai_2005,zeitouni1992generalized}}.
Thus, new methods are required with low complexity and satisfactory  non-asymptotic performance.
	In general, for two-sided hypothesis testing, some restrictions have to be added to the locally most powerful (LMP) approach in order to obtain a valid test.
Two widely-used restrictions are: 1)  invariance tests, which results in the 
	LMP invariant test \cite{6376186,lmpid}; and 2)  unbiasedness  in the sense of tests where the probability of detection is greater than (or equal to) the probability of false alarm \cite{Lehmann},  which results in the
	LMP unbiased (LMPU) test \cite{LPMUbook}.
	In this paper,  the second approach has been adopted in order to obtain a valid test  for composite two-sided hypothesis testing.

In this paper, we consider the problem of the detection of a frequency deviation from a nominal value for sinusoidal signals with unknown complex amplitudes.  
Since the hypotheses are assumed to be close, 
our goal is  to develop a low-complexity detector for the detection of small  changes in the frequency.
To this end, we present  the theoretical concept of the  {\em{generalized locally most powerful unbiased}} (GLMPU) detector. The GLMPU detector provides a general local detection approach in the presence of unknown nuisance  parameters. Similar to the concept of the GLRT \cite{Poor}, the GLMPU  detector is obtained by  substituting the ML estimators of the nuisance parameters into the LMPU  test. When the estimation error of the nuisance parameters is small, the GLMPU test is expected
to be close to the LMPU test. The GLMPU test is the two-sided version of the 
	GLMP for one-sided local hypothesis testing  in the presence of unknown nuisance parameters described in   \cite{marohn2001comment,pmutir}.
	We derive closed-form expressions of the LMPU and GLMPU
tests for  local frequency deviation detection.
We provide simulation results in practical
settings and demonstrate  that the proposed GLMPU and LMPU tests outperform the GLRT methods for a small false alarm probability.
Furthermore, the computational complexity of the LMPU and GLMPU tests is lower than that of the GLRT, since it does require estimation of the frequency.    


\section{Problem formulation: Detection of small frequency deviations}
\label{Model_Definition}
	We consider a binary hypothesis testing problem in which $M$ sensors collaborate to detect the presence of a frequency deviation in a sinusoidal signal. This  hypothesis
	testing problem is formulated as follows: 
	\begin{align}\label{eq1}
	\left\{\begin{array}{l}
	\mathcal{H}_{0}:~x_m[n]= 
	A_{m}e^{j\gamma n} + w_m[n] \\
	\mathcal{H}_{1}:~x_m[n]= 
	A_{m}e^{j\gamma \frac{\omega_0+\Delta}{\omega_0} n} + w_m[n]\end{array}\right. ,~ 
	 n=0,\ldots,N-1 ~{\text{and}} ~  m=1,\ldots,M,
	\end{align}
where $x_m[n]$ is the observation at time $n$ measured by the $m$th sensor, $\gamma $ represents the sampling angle, and the sequence $\{w_m[n]\}^{N-1}_{n=0}$ is an independent complex circularly symmetric zero-mean Gaussian noise sequence with known variances $\sigma_m^2$,  $m=1,\dots,M$.
The $m$th complex amplitude  is denoted by $A_{m}\in{\mathbb{C}}$, $ m=1,\ldots,M$, $\omega_{0}$ is the {\em{known}} nominal system frequency under normal conditions, and $\Delta$ is the {\em{unknown}} frequency deviation.
That is,  the observations are sampled versions of a sinusoidal signal  sampled $\frac{2\pi}{ \gamma}$ times per
cycle of the nominal frequency, $\omega_0$, where the signal frequency is $\omega_0$ and $\omega_0+\Delta$ under hypotheses
$\mathcal{H}_{0}$ and $\mathcal{H}_{1}$, respectively.
The goal is to detect whether there is a nonzero frequency deviation, $\Delta$, which is  assumed to be small,
based on the $M$  observation vectors, $\xvec_m\define [x_m[0],\ldots,x_m[N-1]]^T$, $m=1,\ldots,M$, in the presence of unknown nuisance parameters,
	$A_1,\ldots,A_{M}$.  
The model in \eqref{eq1}

For this case, the likelihood ratio test (LRT), which assumes the knowledge of the unknown parameters $A_1,\ldots,A_{M}$, and $\Delta$, is given by (see, e.g. Subsection 3.2 in \cite{wiley}):
\beqna\label{eq4}
T_{\text{LRT}}(\xvec) =
\sum_{m=1}^{M}\frac{{\text{Re}}\left\{A_{m}\xvec_{m}^{H}
\svec(\omega_{0}+\Delta)\right\}}{\sigma_m^{2}} 
-\sum_{m=1}^{M}\frac{{\text{Re}}\left\{A_{m}\xvec_{m}^{H}
 \svec(\omega_{0}) \right\}}{\sigma_m^{2}},
\eeqna
where $\xvec\define[\xvec_1^T,\ldots,\xvec_M^T]^T$,  
 ${\text{Re}}\{\cdot\}$ and  ${\text{IM}}\{\cdot\}$ denote the real and imaginary parts of its argument, and 
\be \label{eq1_a}
\svec(\omega) \define  \left[ 1, e^{j\gamma \frac{\omega}{\omega_0}  }, e^{j\gamma \frac{\omega}{\omega_0} 2} , \dots , e^{j\gamma \frac{\omega}{\omega_0} (N-1)} \right]^{T},~\forall\omega\in{\mathbb{R}}.  
\ee
The left and right terms on the r.h.s. of \eqref{eq4} are associated with the log-likelihood functions under hypothesis 
	$\mathcal{H}_1$ and $\mathcal{H}_0$, respectively.
	
The GLRT
replaces the unknown parameters in the two likelihoods in the  LRT  with the  associated ML estimator of the unknown parameters under each hypothesis.
In particular, the ML estimator of the frequency deviation, $\Delta$, under hypothesis $\mathcal{H}_1$ is given by \cite{singletone}:
\begin{align} \label{eq9}
\hat{\Delta} = \arg\max_{\alpha \in [-\frac{\omega_0\pi}{\gamma},\frac{\omega_0\pi}{\gamma})} \frac{1}{N}\sum_{m=1}^{M} \frac{|\svec^{H}(\omega_{0}+\alpha)\xvec_{m}|^2}{\sigma_m^2},
\end{align}
where we restrict the estimates of $\Delta$ to be in 
$[-\frac{\omega_0\pi}{\gamma},\frac{\omega_0\pi}{\gamma}) $, in order to avoid ambiguities. 
The ML estimators of $A_{m}$ under hypotheses $\mathcal{H}_1$ and $\mathcal{H}_0$  are given by
$
\hat{A}_{m}^{(1)} 
= \frac{1}{N}\svec^{H}(\omega_{0}+\hat{\Delta}) \xvec_{m}$
and
$
\hat{A}_{m}^{(0)} = \frac{1}{N}\svec^{H}(\omega_{0}) \xvec_{m}$,  for $m=1,\ldots,M$,
respectively.
By substituting (\ref{eq9})   and $\hat{A}_{m}^{(1)}$ into the left  term on the r.h.s. of (\ref{eq4}), 
and substituting $\hat{A}_{m}^{(0)}$ in the right term, we obtain the the GLRT  for the considered problem:
\beqna
\label{GLRT}
T_{\text{GLRT}}(\xvec) 
 =\frac{1}{N}
\sum_{m=1}^{M}\frac{|\xvec_{m}^{H} \svec(\omega_{0}+\hat{\Delta})|^2-
|\xvec_{m}^{H} \svec(\omega_{0})|^2}{\sigma_m^{2}}  .
\eeqna
	
	The GLRT in \eqref{GLRT} has two  fundamental drawbacks: 1) it requires the computation of the ML estimator of the frequency deviation, $\Delta$, from  (\ref{eq9}), which is based on a  search approach and, thus, suffers from high computational complexity and long runtime for real-time applications (see, e.g. \cite{nielsen,singletone}); and 
	2)  when $\Delta$ has small values, the alternative is close to the null hypothesis
	and the GLRT performance may degrade and be outperformed by local detectors.  
	As an alternative to the GLRT, in the following we construct a new concept of the GLMPU test, which has merit in terms of low computational complexity   and  high detection performance, especially for small deviations of the detected parameter.


\section{GLMPU test for local detection with nuisance parameters}
\label{GLMP_sec}
In many cases,  the optimal  UMP test   does not exist. For these cases, the LMP test yields the maximum probability of detection  for weak  signals that are near the local value   \cite{Poor,Kay_detection,Lehmann}. However, the LMP and LMPU tests do not involve  unknown parameters, except for the local parameter. In this section, we propose the  GLMPU test, which is a generalization of the  LMPU test for two-sided hypothesis testing regarding a local parameter (i.e. detection of weak signals) with additional, unknown nuisance parameters. 
We  develop the GLMPU test  for  general  local detection in the presence of  nuisance parameters in Subsection \ref{general_GLMP}. Then,  we derive the LMPU and the GLMPU tests for the special case of the detection of frequency deviation in Subsection \ref{special_case}.

\subsection{GLMPU test}
\label{general_GLMP}
 We consider the following general two-sided composite hypothesis testing:
    \begin{align}
    \label{eqgen}
   	\left\{\begin{array}{l}
    \mathcal{H}_{0}: \xvec \sim f(\xvec; \theta_l, \thetavec_{n}),~\theta_l=\theta_0  \\
    \mathcal{H}_{1}:  \xvec \sim f(\xvec; \theta_l, \thetavec_{n}), ~ \theta_l \ne \theta_{0}
    \end{array}\right.,
    \end{align}
    where,
    with slight abuse of notation, $\xvec\in \Omega_\xvec$ in this subsection is the observation vector for the general case, where $\Omega_\xvec$ is the observation space. 
    The pdf of  $\xvec$ under both hypotheses,
    $f(\xvec; \theta_l, \thetavec_{n})$,
    is assumed to be a continuous and twice differentiable function with respect to (w.r.t.) the local parameter, $\theta_l\in{\mathbb{R}}$,
    for any  unknown nuisance parameter vector,  
    $\thetavec_n \in \mathbb{C}^K$.  
Our goal here is to implement the LMPU test for the model in (\ref{eqgen}),
by maximizing the probability of detection for a given probability of false alarm for small deviations from the null hypothesis  around a boundary value of a \emph{local} parameter, $\theta_l=\theta_0$, i.e. in the local, open neighborhood,  
    \begin{align}\label{eqOmegaH}
        \Omega_\delta\define \{\theta_l\in{\mathbb{R}}, |\theta_l - \theta_{0}|<\delta\}.
    \end{align}
  Since the hypothesis testing in \eqref{eqgen} includes two-sided alternatives, we add the unbiasedness as an extra condition for the development of the LMP test.

 A general non-random test, based on the observation vector, $\xvec$, can be defined as
    \begin{align}
    \label{general_test}
    \Phi(\xvec) \define 
    	\left\{\begin{array}{lr}
    1, & \xvec \in \mathcal{S}_{1} \\
    0,& \xvec \notin \mathcal{S}_{1}
    \end{array}\right.,
    \end{align}
    where $\mathcal{S}_{1}\subset \Omega_\xvec$  is  
    the rejection region, which includes values of the test statistic, $\Phi(\xvec)$, that lead to rejection of 
    $\mathcal{H}_{0}$ (acceptance of  $ \mathcal{H}_{1}$). 
    Thus,  the probability of detection and the probability of the false alarm 
    for  the hypothesis testing in \eqref{eqgen} by using the  general test in \eqref{general_test}
    are
    \beqna\label{eqPro}
    P_{D}(\theta_l,\thetavec_{n})=   {\rm{E}}_{\theta_l}  [\Phi(\xvec)]=\int_{\Omega_\xvec} \Phi(\xvec) f(\xvec; \theta_l,\thetavec_n) \ud \xvec ~ ,  \theta_l \ne \theta_0
    \eeqna
    and
    \be
    \label{PFA_def}
    P_{FA} (\theta_0, \thetavec_n)={\rm{E}}_{\theta_l}  [\Phi(\xvec)]=\int_{\Omega_\xvec} \Phi(\xvec) f(\xvec; \theta_l,\thetavec_n) \ud \xvec  , ~ \theta_l = \theta_0,
    \ee
    respectively.
    
    The requirements on the desired test are:
    \begin{enumerate}
        \item Size $\alpha$ - By using \eqref{PFA_def}, this requirement can
        be written as
    \be
    \label{Pfa}
   P_{FA}(\theta_{0},\thetavec_{n})=\alpha,
   \ee
   where $\alpha\in[0,1]$.
    \item Locally unbiasedness - 
 A  test $\Phi$  of size $\alpha$ is  a locally unbiased test for the hypothesis testing in \eqref{eqgen} if
   \eqref{Pfa} is satisfied  and, in addition, 
    \be
   \label{unbiased}
   P_{D}(\theta_l,\thetavec_{n}) \geq \alpha,~\forall \theta_l\in \Omega_\delta,
   \ee
    where   the set  $\Omega_\delta$ and $ P_{D}(\theta_l,\thetavec_{n}) $  are defined in \eqref{eqOmegaH} and
    \eqref{eqPro}, respectively.
    \end{enumerate}

The LMPU test is the test that maximizes the probability of detection under the constraints in \eqref{Pfa} and \eqref{unbiased}. The following theorem presents the explicit test as a function of the likelihood function.
    \begin{theorem}
    \label{Th_LMPU} (LMPU test)
    The LMPU test
    for a known parameter vector, $\thetavec_n$,   is given by:
   \beqna \label{eqlmpu}
    T_{\text{LMPU}}(\xvec) =
    \frac{\partial^2 \log f(\xvec; \theta_l, {\thetavec}_{n})}{\partial \theta_l^2}|_{\theta_l=\theta_{0}}  +
    \left( \frac{\partial \log f(\xvec; \theta_l, {\thetavec}_{n})}{\partial \theta_l}|_{\theta_l=\theta_{0}}\right)^2 -
    \tilde{\kappa}_{1} -\tilde{\kappa}_{2} \frac{\partial \log f(\xvec; \theta_l, {\thetavec}_{n})}{\partial \theta_l}|_{\theta_l=\theta_{0}},
    \eeqna
    where $\tilde{\kappa}_{1}$ and $\tilde{\kappa}_{2} $ are determined by the 
    $\alpha$-size 
    and  unbiasedness constraints from  \eqref{Pfa} and \eqref{unbiased}, respectively.
    \end{theorem}
\begin{proof}
The proof is along the path of the development of the LMPU test (see, e.g. pp. 369-370 in \cite{LPMUbook}). For the sake of completeness  it appears in the Appendix.
 \end{proof}
 It can be shown that, under  mild conditions \cite{raul,ghosh}, the coefficients $\tilde{\kappa}_{1}$ and $\tilde{\kappa}_{2}$ in Theorem \ref{Th_LMPU}  exist and can be determined uniquely from the constraints described in Theorem \ref{Th_LMPU}.
In  particular, in Subsection 3.6 in \cite{Lehmann} it is shown that the coefficients $\tilde{\kappa}_{1},\tilde{\kappa}_{2} $ should be non-negative in order to satisfy \eqref{unbiased} and  \eqref{Pfa}.
 In the general case, these coefficients  are functions of the problem parameters, $\theta_l$ and $\thetavec_n$.   
    
    Similar to the GLRT concept  \cite{Poor,Kay_detection},
    we propose in this paper
    the novel GLMPU  test, which is derived by replacing nuisance parameters in the LMPU test statistic from Theorem \ref{Th_LMPU} by their corresponding ML estimators calculated at the point $\theta_l = \theta_0$, as described in the following definition.
\begin{definition}
\label{GLMPdef} (GLMPU test)
The GLMPU test for the hypothesis testing problem in \eqref{eqgen} with unknown nuisance parameters, $\thetavec_n$, is obtained by substituting  the following ML estimator of $\thetavec_n$:
	  \begin{align}\label{option1}
    \hat{\thetavec}_{n} = \arg\max_{\thetavecsmall_{n}\in{\mathbb{C}}^M} \log f(\xvec; \theta_l,\thetavec_{n} )|_{\theta_l = \theta_{0}},
    \end{align}
	 into the LMPU test in \eqref{eqlmpu}.
	 Thus, the GLMPU test is given by: 
	   \beqna \label{eqGlmpu}
        T_{\text{GLMPU}}(\xvec)	
        =\frac{\partial^2 \log f(\xvec; \theta_l, \hat{\thetavec}_{n})}{\partial \theta_l^2}|_{\theta_l=\theta_{0}} + 
    \left( \frac{\partial \log f(\xvec; \theta_l, \hat{\thetavec}_{n})}{\partial \theta_l}|_{\theta_l=\theta_{0}}\right)^2 - \kappa_{1} - 
     \kappa_{2} \frac{\partial \log f(\xvec; \theta_l, \hat{\thetavec}_{n})}{\partial \theta_l}|_{\theta_l=\theta_{0}},
    \eeqna
    where  $\kappa_{1}\geq 0$ and $\kappa_{2}\geq 0$ are determined from the following conditions for the chosen $\alpha\in[0,1]$:
    \begin{align}\label{eqcondhat1}
   &P_{FA}(\theta_{0},\hat{\thetavec}_n)=\alpha, \\
   \label{eqcondhat2}
   &  P_{D}(\theta_l,\hat{\thetavec}_n) \geq \alpha,~\forall \theta_l\in \Omega_\delta.
    \end{align}
    \end{definition}
  It should be noted that the conditions  in \eqref{eqcondhat1} and \eqref{eqcondhat2} are obtained by substituting the estimator $\hat{\thetavec}_n$ from \eqref{option1} instead of $\thetavec_n$  in conditions \eqref{Pfa} and \eqref{unbiased}. 
   In the general case, the coefficients ${\kappa}_{1}$ and ${\kappa}_{2}$ in Definition \ref{GLMPdef} are functions of  $\theta_l$ and of the estimator, $\hat{\thetavec}_n$. 
   In practice,  the calibration of the coefficients $\kappa_{1}$ and $\kappa_{2}$  can be performed  offline by a series of experiments for typical values and under some assumptions.

  In terms of computational complexity, the  advantage of the proposed GLMPU test is evident when the nuisance parameters are easy to estimate by the ML estimator, while the ML estimator of the local parameter requires a search approach.
  Similar to the derivation of the LMPU test in the Appendix, the GLMPU test can be obtained as a solution to a similar optimization problem as \eqref{OP} in the Appendix, by replacing the unknown vector $\thetavec_n$ with its ML estimator, $\hat{\thetavec}_n$, from \eqref{option1}.
  
 In the rare cases where the LMPU test is
independent of nuisance parameter vector, $\thetavec_n$, the GLMPU test from Definition \ref{GLMPdef} coincides with the LMPU test from Theorem \ref{Th_LMPU}. For the general case, 	since the ML estimator asymptotically converges to the true value of the estimated parameters,
the proposed GLMPU test is expected to achieve the performance of the LMPU test asymptotically, i.e. for a sufficient number of measurements and/or for a high signal-to-noise ratio (SNR).
 Another important special case is that of one-sided hypothesis testing, where the neighborhood in \eqref{eqOmegaH}
	is replaced by
	 $
        \Omega_\delta^{os}\define \{\theta_l\in{\mathbb{R}}, 0\leq \theta_l - \theta_{0}<\delta\}$.
	For this case, the derivation of the GLMPU test for the neighborhood $\Omega_\delta^{os}$ will result in  the recent one-sided GLMP
	\cite{pmutir}:
	\begin{align} \label{eqGLMP1sided}
          T_{GLMP}(\xvec) = \frac{\partial}{\partial \theta_l} \log f(\xvec,\theta_l,\hat{\thetavec}_n)|_{\theta_l=\theta_0},
    \end{align}
which can be interpreted as the rightmost term on the GLMPU in \eqref{eqGlmpu}.

\subsection{LMPU and GLMPU tests for frequency deviation detection}
\label{special_case}
In this subsection 
we  develop the LMPU test from Theorem \ref{Th_LMPU} and the GLMPU test from Definition \ref{GLMPdef}  for the special case of the detection of frequency deviation of sinusoidal signals with unknown complex amplitudes,  described  in Section \ref{Model_Definition}. 
In this case, the nuisance  parameter vector is  $\thetavec_{n} = [A_1,\ldots,A_{M} ]^{T}\in{\mathbb{C}}^M$ and the local parameter is $\theta_l=\Delta$ with the associated 
  boundary value of  $\theta_0=0$.  Thus, the  open neighborhood  
   from \eqref{eqOmegaH} in this case is $\Omega_\delta= \{\Delta\in{\mathbb{R}}, |\Delta |<\delta\}$.

The log-likelihood function for this model (after removing constant terms w.r.t. $\Delta$ ) is given by:
\beqna
\label{like_freq}
\log f(\xvec;\Delta,{\thetavec}_n)= -\sum_{m=1}^{M} \frac{1}{\sigma_{m}^{2}}||\xvec_m -A_m \svec(\omega_0 + \Delta) ||^{2} ,
\eeqna
where $\svec(\omega)$ is defined in \eqref{eq1_a}. In addition,   $\Delta=0$  under the null hypothesis, and  $\Delta \ne 0$ under the alternative.
The first- and second-order derivatives of the log-likelihood function in \eqref{like_freq} w.r.t. the local parameter, $\theta_l=\Delta$, are
	\begin{align} \label{lmpdef}
	\frac{\partial \log f(\xvec;\Delta,\thetavec_{n})}{\partial \Delta}|_{\Delta=0} =-
	\sum_{m=1}^{M}\frac{\gamma}{\omega_0\sigma_m^{2}} {\text{Im}}\left\{A_{m}\xvec_{m}^{H}  \Dmat_N \svec(\omega_{0}) \right\}
	\end{align}
	and
	\begin{align} \label{lmpdef1}
	\frac{\partial^2 \log f(\xvec;\Delta,\thetavec_{n})}{\partial \Delta^2}&|_{\Delta=0}=-\sum_{m=1}^{M}\frac{\gamma^2}{\omega_0^2\sigma_m^{2}}{\text{Re}}\left\{A_{m}\xvec_{m}^{H} \Dmat_N\Dmat_N \svec(\omega_{0}) \right\},
	\end{align} 
respectively, where  $\Dmat_N$ is a diagonal matrix with the diagonal elements  $[\Dmat_N]_{n} =n-1$, $\forall n=1,\ldots,N$.  By substituting \eqref{lmpdef} and \eqref{lmpdef1} in 	 
the LMPU test  in (\ref{eqlmpu}), one obtains 
	\begin{align}\label{eqLMPU}\nonumber
	    T_{LMPU}(\xvec) =& -\sum_{m=1}^{M}\frac{\gamma^2}{\omega_0^2\sigma_m^{2}}{\text{Re}}
	    \left\{A_{m}\xvec_{m}^{H} \Dmat_N \Dmat_N \svec(\omega_{0}) \right\} + \left(  \sum_{m=1}^{M}\frac{\gamma}{\omega_0 \sigma_m^{2}}{\text{Im}}\left\{A_{m}\xvec_{m}^{H}  \Dmat_N \svec(\omega_{0}) \right\} \right)^2 \\
	    &- \tilde{\kappa}_{1} + \tilde{\kappa}_{2} \sum_{m=1}^{M}\frac{\gamma}{\omega_0 \sigma_m^{2}}{\text{Im}}\left\{A_{m}\xvec_{m}^{H}  \Dmat_N \svec(\omega_{0}) \right\}.
	\end{align}
The LMPU test in \eqref{eqLMPU}  is a function of
the unknown parameters, $A_1,\ldots,A_m$, and, thus, the LMPU test is useful only as a benchmark on the performance of practical estimators.

 In order to obtain the GLMPU test, the unknown nuisance parameter vector, $\thetavec_{n}$, is  replaced by its  ML estimator at the point $\Delta = 0$, as defined in (\ref{option1}). The ML estimator of $A_{m}$ (i.e. the 
$m$th component  of  $\hat{\thetavec}_{n}$) is given by:
\begin{align} \label{eqlmpX}
[\hat{\thetavec}_{n}]_{m}=\hat{A}_{m} = \frac{1}{N} \svec(\omega_{0})^{H}\xvec_{m},~ \forall m=1,\ldots,M.
\end{align}
It can be seen that for this case the ML estimator of the nuisance parameter vector in \eqref{eqlmpX} is a linear function of the observation vector, $\xvec$, and does not require a search approach.
By substituting \eqref{eqlmpX} in (\ref{eqLMPU}) and replacing $\tilde{\kappa}_i$ by $\kappa_i$, $i=1,2$, we obtain that
the GLMPU test  in our case is
\begin{align}\label{eqGLMPU}\nonumber
  T_{GLMPU}(\xvec)=  &-\sum_{m=1}^{M}\frac{\gamma^{2}}{N \omega_{0}^{2} \sigma_m^{2}} {\text{Re}}\left\{\xvec_{m}^{H}  \Dmat_N \Dmat_N \svec(\omega_{0})\svec^{H}(\omega_{0}) \xvec_{m} \right\} +\left(\sum_{m=1}^{M}\frac{\gamma}{N \omega_0 \sigma_m^{2}}\text{Im}\left\{\xvec_{m}^{H}  \Dmat_N\svec(\omega_{0})\svec^{H}(\omega_{0}) \xvec_{m} \right\} \right)^{2} \\
  & - \kappa_{1} +\kappa_{2} \sum_{m=1}^{M}\frac{\gamma}{N \omega_0 \sigma_m^{2}}\text{Im}\left\{\xvec_{m}^{H}  \Dmat_N\svec(\omega_{0})\svec^{H}(\omega_{0}) \xvec_{m} \right\}.
\end{align}

The computational complexity of the GLMPU test in \eqref{eqGLMPU} is lower than that of the GLRT in \eqref{GLRT}, since, as can be seen in \eqref{eqGLMPU}, the test is based on  quadratic transformations of the observation vector, $\xvec$, and does not require a search approach in order to estimate the frequency deviation, $\Delta$.
In particular, in order to calculate $\hat{\Delta}$ from \eqref{eq9},  we need to define the grid for searching over the parameter $\alpha \in [-\frac{\omega_0\pi}{\gamma},\frac{\omega_0\pi}{\gamma})$. Then, 
for each value in the grid $\alpha$, the computational complexity  for calculating the 
 vector multiplications $|\svec^{H}(\omega_0+ \alpha) \xvec_{m}|^{2}$, $\forall m=1,\ldots,M$  costs  $2N M$ flops (see  Appendix C in \cite{time}) in addition to the summation  over $M$, which requires $M-1$ flops. Thus, the total number of flops for the computation of the ML estimator is $N_{\alpha} (M(2N+1)-1)$,	
where $N_{\alpha}$ denotes the number of points in the chosen grid of $\alpha$.
When $N_{\alpha}$ is larger, one can achieve better estimation performance of $\Delta$, which results in better detection performance of the GLRT, but at the cost of increased computational complexity.
In addition to the search approach,  the GLRT  from \eqref{GLRT} requires multiplications and summation with  ${\cal{O}}(MN)$ additional flops.
In contrast to the GLRT, which has a computational complexity of the order of  ${\cal{O}}(N_{\alpha} MN+MN)$,	the computational complexity of the GLMPU test from \eqref{eqGLMPU}  is based on linear operators, and is   $ {\cal{O}}(MN)$.

\vspace{-0.25cm}
\section{Simulations}
	\label{simulations_sec}
In this section we  evaluate the performance of the proposed GLMPU test and compare it with the performance of the LRT, GLRT, and LMPU tests  in terms of detection performance and computational complexity. Our simulations are based on the important application of power systems, based on  a measurement model of Phasor Measurement Units (PMUs)   \cite{PMUbook,kcl}, where it is known that the frequency deviations in electrical networks  are  small (see, e.g.  Table 1 in \cite{power}). 
PMUs have been increasingly
deployed in wide area transmission networks and are able to
accurately measure voltage and current phasors at a high frequency   with synchronized
time stamps.
We consider a single PMU that can be represented by the model in \eqref{eq1}, where in this case  $\{ A_m\}_{m=1}^{M-1}$ represent the phasors of the   currents of $M-1$ transmission lines and $A_M$ is the voltage phasor at a specific node (the detailed model is described in    \cite{kcl}).
In the simulations below,
we set
$M=6$, $\thetavec_n =[
	1,1 e^{j\frac{\pi}{3}},\sqrt{3} e^{-j\frac{5\pi}{6}},1,1e^{j\pi},1
	]^T $, and the nominal-frequency  $\omega_0=2\pi\cdot 60 $ $[\frac{rad}{sec}]$.  The sampling rate is $N=48$ samples per cycle of the nominal power frequency. We assume equal SNRs, and define  $SNR \define \frac{| A_{m}|^2}{\sigma_m^2}=0 [dB]$, $\forall m=1,\ldots,M$, unless otherwise specified.

 Figures \ref{fig8} and \ref{fig8a}  presents the performance of the GLRT and LMPU test when the nuisance parameter vector, $\thetavec_n$, is {\em{known}}. 
The GLRT in this case is obtained by substituting the  ML estimator  $\hat{\Delta}$ from \eqref{eq9} in \eqref{eq4}, with the true values of $A_m$, $m=1,\ldots,M$.
The GLMPU test in this case is reduced to the LMPU test  from \eqref{eqGLMPU}, where we used $\kappa_{2} = 0$ and tune the value of $\kappa_{1}$   such that the constraints  \eqref{Pfa} and \eqref{unbiased}  hold. It can be seen that both the tests, GLMPU and GLRT, are unbiased,  since their probability of detection is greater than or equal to the probability of false alarm for any $\Delta$. Additionally,  the LMPU test outperforms the  GLRT for any value of  $\Delta$ in this scenario, in the sense of probability of detection for any tested false alarm probability, $P_{FA}(\Delta,\thetavec_{n}) = \alpha = 0.01,0.05,0.1$. In addition, from Fig \ref{fig8a}, it can be seen that  the performance of LMPU test is always better than or equal to the performance of the GLRT, in terms of probability of detection.

In Fig. \ref{fig7} the probability of detection of the GLRT and GLMPU test are presented when the nuisance parameter vector, $\thetavec_n$, is {\em{unknown}}, versus the local parameter, $\Delta$, for  $P_{FA}(\Delta,\thetavec_{n})=0.01,0.05,0.1$. It can be seen that the probability of detection of the GLMPU test is higher than the probability of detection of the  GLRT  in the neighborhood at $\Delta = 0$ for any value of false alarm probability, and that the gap between the tests is larger when the probability of false alarm is smaller. 
Similarly, in Fig. \ref{f11} it can be seen that the performance of GLMPU test is better than GLRT in sense of the probability of detection verses SNR for any tested false alarm probability, $P_{FA}(\Delta,\thetavec_{n}) = \alpha = 0.01,0.05,0.1$. 
In Fig. \ref{fig12} the receiver operating characteristic (ROC) curves of the GLRT and GLMPU test are presented for the case where we have a single sample for detection, i.e. for $N=1$.  It can be seen that the GLRT cannot detect frequency changes  based on a single sample, since the ML estimator in \eqref{eq9} requires $N\geq 2$ samples. On the other hand, the GLMPU test  achieves good detection performance in this case for high SNRs.
The probability of detection of the GLMPU test increases when $\Delta$ increases, since it is easier to distinguish between the null and the alternative hypotheses as $\Delta$ increases. The special case of a single sample is  useful for  real-time applications, for accurate and fast change detection of frequency. 


In order to demonstrate the empirical complexity of the proposed methods for different problem dimensions, the average computation time, “run-time”, was evaluated by running the algorithms using Matlab on an Intel Core(TM) $i7-10TH~GEN$ CPU computer, $2.80$ $GHz$. Fig. \ref{f10} shows that the run-time of the GLMPU test  is much shorter than  the run-time of the
GLRT for any  number of measurements, $N$.
It can be seen that
the run-time increases polynomially with the number of measurements and sensors, $N$ and $M$, as expected from the theoretical discussion on computational
complexity at the end of Subsection \ref{special_case}.

	\newpage
	\section*{Figures}
\begin{figure} [H] 
	\begin{center}
		\includegraphics[width=12cm]{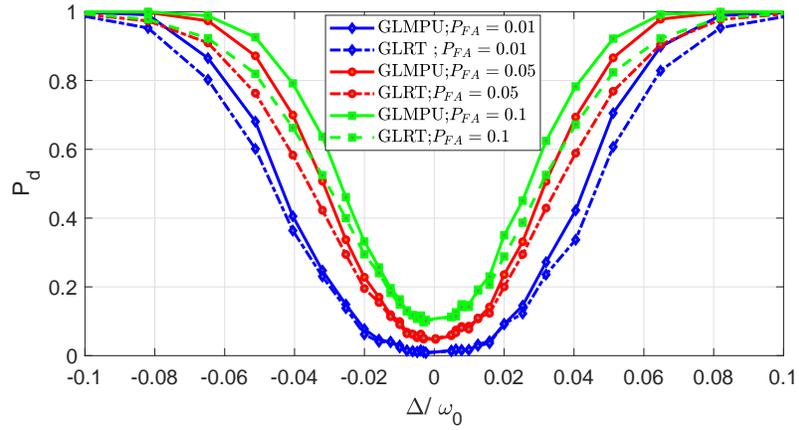}
		\caption{Comparison between GLRT and LMPU test in terms of  probability of detection versus the local parameter $\Delta$, where $\Delta$ is  unknown, $P_{FA}(\Delta,\thetavec_{n})=0.01,0.05,0.1$ and with the parameters  $\omega_0=2\pi*60 [\frac{rad}{sec}]$, $ SNR= 0[dB]$, $N_{\alpha}=60,000$, $M=6$, and $N=48$.}
		\label{fig8}
	\end{center}
\end{figure}
\begin{figure} [H] 
	\begin{center}
		\includegraphics[width=12cm]{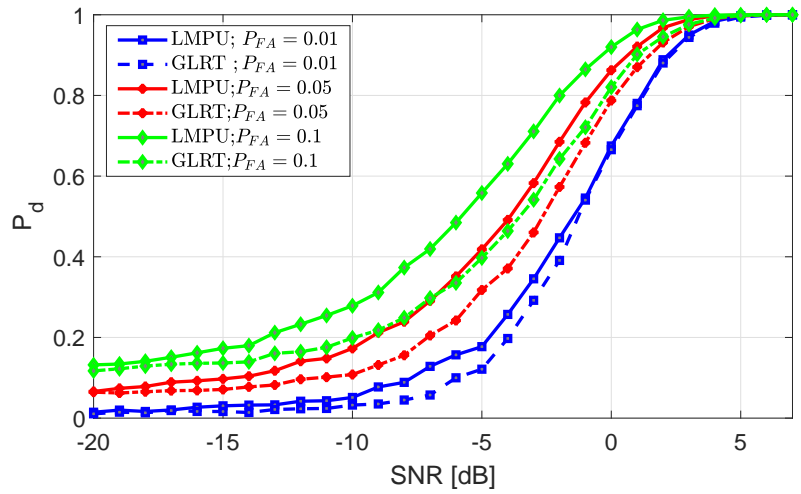}
		\caption{Comparison between GLRT and LMPU test in terms of  probability of detection versus SNR, where $\Delta$ is unknown and the vector $\thetavec_n$ is known. The parameters are setting to be: $\Delta=0.242 \omega_0$, $\omega_0=2\pi*60 [\frac{rad}{sec}]$, $N_{\alpha}=60,000$, $M=6$, and $N=48$.}
		\label{fig8a}
	\end{center}
\end{figure}
\begin{figure} [H] 
	\begin{center}
		\includegraphics[width=12cm]{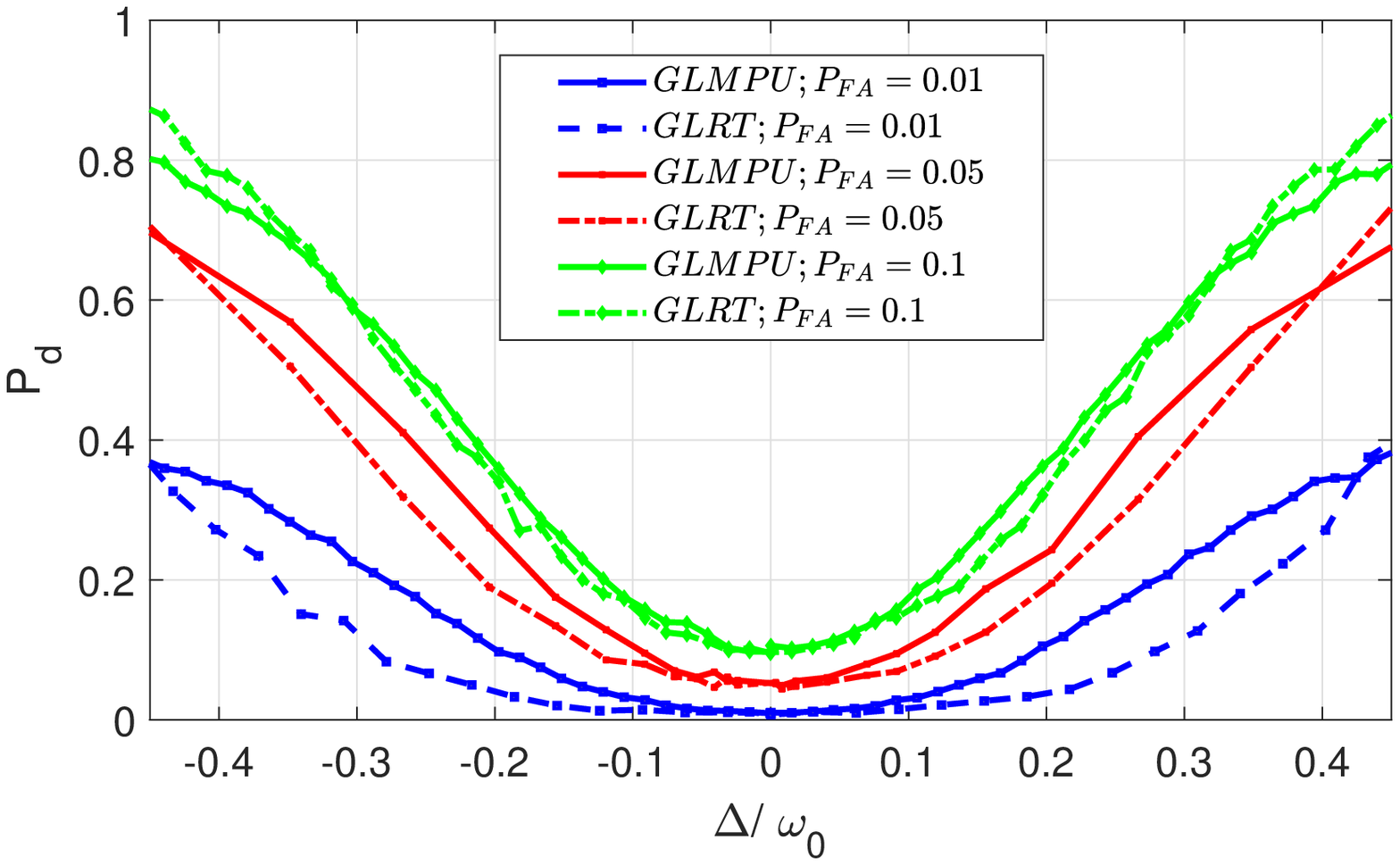}
		\caption{Comparison between GLRT and GLMPU test in terms of  probability of detection versos the local parameter $\Delta$, where $\Delta$ and $\thetavec_n$ are unknown and with the parameters $\omega_0=2\pi*60 [\frac{rad}{sec}]$, $ SNR =0[dB]$, $N_{\alpha}=60,000$, $M=6$, and $N=4$.}
		\label{fig7}
	\end{center}
\end{figure}

\begin{figure} [H] 
	\begin{center}
		\includegraphics[width=12cm]{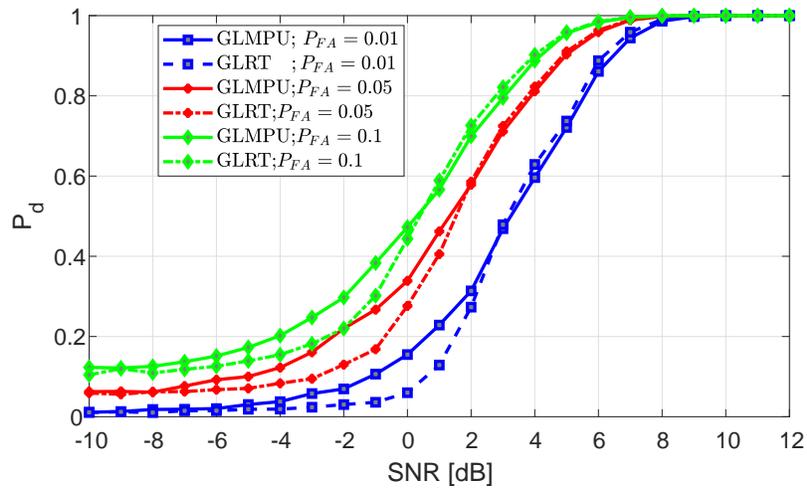}
		\caption{Comparison between GLRT and GLMPU test in terms of  probability of detection versus SNR, where $\Delta$ and $\thetavec_n$ are unknown. The parameters are setting to be: $\Delta=0.242 \omega_0$, $\omega_0=2\pi*60 [\frac{rad}{sec}]$, $N_{\alpha}=60,000$, $M=6$, and $N=4$. }
		\label{f11}
	\end{center}
\end{figure}

\begin{figure} [H] 
	\begin{center}
		\includegraphics[width=12cm]{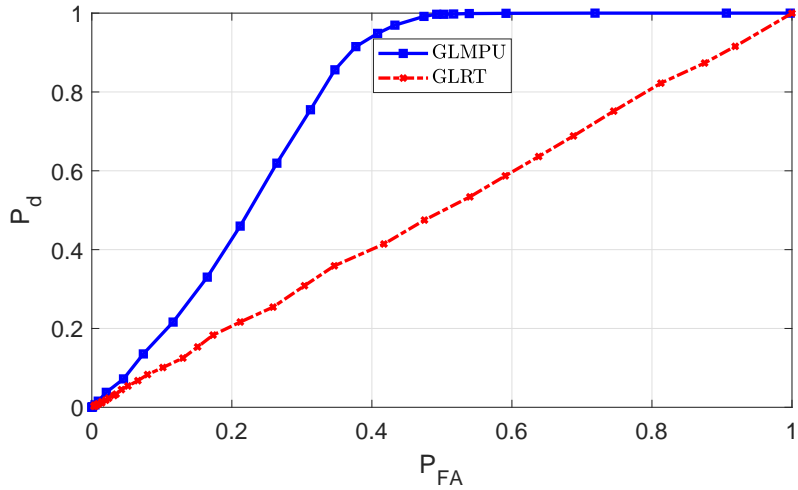}
		\caption{ROC curves of the GLRT and GLMPU test
with the  parameters $M=6$, $N=1$, $\Delta = 0.242 \omega_0$, $ SNR =10[dB]$, and $N_{\alpha}=60,000$, where  $\thetavec_n$ is unknown.}
		\label{fig12}
	\end{center}
\end{figure}


\begin{figure} [H] 
	\begin{center}
		\includegraphics[width=12cm]{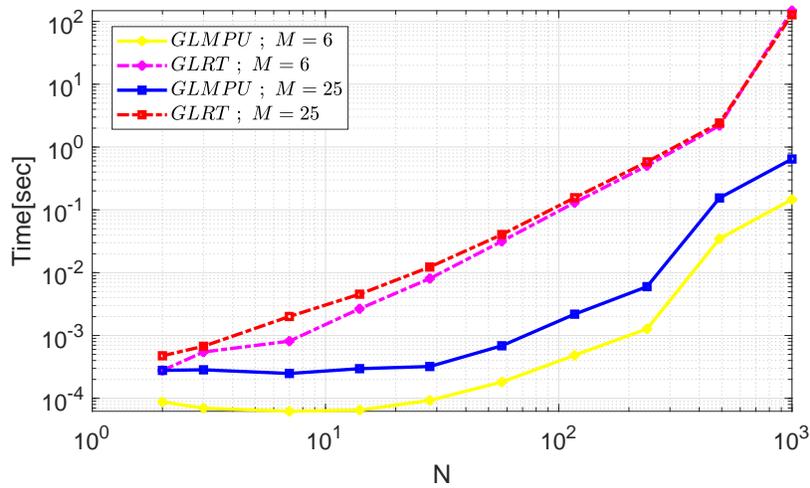}
		\caption{Comparison between GLRT and GLMPU test in terms of run-time versus number of  measurements $N$, where $\thetavec_n$ and $\Delta$ are unknown. Number of search grid $\alpha$ is $N_{\alpha}=60,000$, $M=6$. }
		\label{f10}
	\end{center}
\end{figure}
\section{Conclusion}
\label{conclusion_sec}
In this paper, we 
consider the problem of detecting frequency changes of signals with unknown amplitudes.
We present
a novel detector for the general problem of two-sided hypothesis testing in the presence of unknown nuisance parameters. 
The proposed GLMPU test is designed specifically for the case of close hypotheses w.r.t. a local  parameter-of-interest.
This case
 is of interest in many practical  systems, where detection is performed in low-SNR 
 and/or small sample size scenarios and there are additional nuisance parameters under both hypotheses.
The proposed GLMPU test, as well as the  LMPU test benchmark,  are demonstrated for the special case of frequency-deviation detection of signals with unknown complex amplitudes.
 Simulation results show that the proposed  GLMPU and LMPU tests  exhibit improved performance compared to the GLRT,
  in terms of probability of detection,
for various practical settings.
 In particular, the GLMPU test is robust to a low level of changes and performs well for
  small sample sizes and for a low probability of false alarm regions. In addition, the GLMPU test
has  significantly lower computational complexity than that of the GLRT since it does not require estimation of the frequency. Future research  topics include the asymptotic analysis of the GLMPU test and applications to various detection problems.


\section{Acknowledgments}
This research was partially supported by THE ISRAEL SCIENCE FOUNDATION (grant No. 1173/16) and by the Israeli Ministry of National Infrastructure, Energy and Water Resources.


\appendix
\section{Proof of Theorem 1}
The power function of a general test $\Phi$ from (8) for the hypothesis  testing  in  (6) is  defined by (p. 69 in \cite{Lehmann}):
\begin{align} \label{eqbeta}
    \beta_{\Phi} (\theta_l, \thetavec_n) \define \int_{\Omega_\xvec} \Phi(\xvec) f(\xvec; \theta_l,\thetavec_n)\ud \xvec ~, \theta_l \ne \theta_0.
\end{align}
The pdf of  $\xvec$ under both hypotheses,
    $f(\xvec; \theta_l, \thetavec_{n})$,
    is assumed to be  a continuous and twice differentiable function w.r.t. the local parameter, $\theta_l\in{\mathbb{R}}$, 
    for any  unknown nuisance parameter vector,  
    $\thetavec_n \in \mathbb{C}^K$.  Therefore, the power function $\beta_{\Phi} (\theta_l, \thetavec_n)$ is also a continuous function w.r.t. the local parameter, $\theta_l\in{\mathbb{R}}$, especially at  $\theta_l= \theta_0$. 

     A level-$\alpha$, unbiased test, $\Phi$, is said to be the LMPU  test
    (p. 340 in \cite{Lehmann})
 if,
for any other given level-$\alpha$, unbiased test $\tilde{\Phi}$, there exists $\delta$ 
such that 
\be\beta_{\Phi} (\theta_l, \thetavec_n)  
\geq \beta_{\tilde{\Phi}} (\theta_l, \thetavec_n) ~,\forall \theta_l\in \Omega_\delta,
\ee
  where  $\Omega_\delta$ is defined in (7).
    Thus,  the LMPU test is obtained by maximizing the power function, $\beta_{\Phi}(\theta_l, \thetavec_n) $,
  under the $\alpha$-size and unbiasedness constraints from \eqref{Pfa} and \eqref{unbiased}, respectively, in the neighborhood $\Omega_\delta$. The constraints from \eqref{Pfa} and  \eqref{unbiased}  can be rewritten by using \eqref{eqbeta} as follows:
  \begin{align}\label{eqFAn}
       P_{FA}(\theta_{0},\thetavec_{n})&= \beta_{\Phi} (\theta_l, \thetavec_n)|_{\theta_l=\theta_0} =\alpha,
 \\\label{equ1}
         P_{D}(\theta_l,\thetavec_{n}) &=\beta_{\Phi} (\theta_l, \thetavec_n)|_{\theta_l \ne \theta_0} \geq \alpha,~\forall \theta_l\in \Omega_\delta.
  \end{align}
 Together, these constraints indicate that  $\beta_{\Phi}(\theta_l,\thetavec_{n})$ has a  minimum point at $\theta_l=\theta_{0}$ on the set $\Omega_\delta$.
    Since we assume that the common pdf, $f(\xvec; \theta_l, \thetavec_{n})$, is twice differentiable in the local neighborhood of $\theta_l$ for any $\thetavec_n\in{\mathbb{C}}^K$,
     the constraint in  \eqref{equ1}  can be replaced by   the stationary condition  \beqna
   \label{op3}
   \left.\frac{\partial \beta_{\Phi}(\theta_l,\thetavec_{n})}{\partial \theta_l}\right|_{\theta_l=\theta_0}=0,
      \eeqna together with the condition
      \beqna
    \label{op3_2}
   \left.\frac{\partial^2 \beta_{\Phi}(\theta_l,\thetavec_{n})}{\partial \theta_l ^2}\right|_{\theta_l=\theta_0}>0.
   \eeqna
 Therefore, by concluding the constraints in \eqref{eqFAn}, \eqref{op3}, and \eqref{op3_2}, the LMPU test can be obtained from the solution of the following constrained optimization problem:
     \begin{align} \label{OP1}
    \max_{\Phi(\xvec)} \beta_{\Phi}(\theta_l,\thetavec_{n}) ~
    s.t.~ \begin{cases}
   \beta_{\Phi}(\theta_{0},\thetavec_{n})= \alpha \\
    \left.\frac{\partial \beta_{\Phi}(\theta_l,\thetavecsmall_{n})}{\partial \theta_l}\right|_{\theta_l=\theta_0}=0
      \\
   \left.\frac{\partial^2 \beta_{\Phi}(\theta_l,\thetavecsmall_{n})}{\partial \theta_l ^2}\right|_{\theta_l=\theta_0}>0
, ~ \forall \theta_l \in \Omega_\delta
    \end{cases}.
    \end{align}
    Under the assumptions of differentiability,
   the Taylor series expansion  of the power function, $\beta_{\Phi}(\theta_l,\thetavec_{n})$, around  $\theta_l=\theta_{0}$ is given by:
    \begin{align} \label{OP2}
    \beta_{\Phi}(\theta_l,\thetavec_{n})
    &= \beta_{\Phi}(\theta_{0},\thetavec_{n}) + (\theta_l - \theta_{0}) \left.\frac{\partial 	 \beta_{\Phi}(\theta_l,\thetavecsmall_{n})}{\partial \theta_l}\right|_{\theta_l = \theta_{0}}
  +
     \frac{1}{2}(\theta_l - \theta_{0})^{2}\left. \frac{\partial^{2} 	\beta_{\Phi}(\theta_l,\thetavecsmall_{n})}{\partial \theta_l^{2}}\right|_{\theta_l = \theta_{0}} 
     +O(\delta^2)
     \nonumber\\
     &=\alpha+
     (\theta_l - \theta_{0})^{2} 
     \left.\frac{\partial^{2} 	\beta_{\Phi}(\theta_l,\thetavecsmall_{n})}{\partial \theta_l^{2}}\right|_{\theta_l = \theta_{0}} +O(\delta^2),
    \end{align}
    where the last equality is obtained by substituting the constraint on the false alarm probability from
    \eqref{eqFAn} and the unbiasedness constraint from \eqref{op3}.
    Thus,  according to \eqref{OP2}, in order
   to obtain the
highest power, $\beta_{\Phi}(\theta_l,\thetavec_{n})$, for a given $\alpha$ and  $\thetavec_n$, we need to maximize
the second order term, $ \left.\frac{\partial^{2} 	\beta_{\Phi}(\theta_l,\thetavecsmall_{n})}{\partial \theta_l^{2}}\right|_{\theta_l = \theta_{0}}$
 for both $\theta_l > \theta_{0}$ and $\theta_l < \theta_{0}$, and this leads to the LMPU test.
 In addition, under the constraint $\beta_{\Phi}(\theta_{0},\thetavec_{n})= \alpha$, \eqref{OP2} implies that the constraint in \eqref{op3_2} is redundant for the maximization of $\beta_{\Phi}(\theta_l,\thetavec_{n})$ (which is always equal to or larger than $\alpha$).
 Thus, the maximization in \eqref{OP1} is equivalent to the following optimization:
     \begin{align} \label{OP1_2}
    \max_{\Phi(\xvec)} \left.\frac{\partial^2 \beta_{\Phi}(\theta_l,\thetavec_{n})}{\partial \theta_l ^2}\right|_{\theta_l=\theta_0} ~
    s.t.~ \begin{cases}
   \beta_{\Phi}(\theta_{0},\thetavec_{n})= \alpha \\
    \left.\frac{\partial \beta_{\Phi}(\theta_l,\thetavec_{n})}{\partial \theta_l}\right|_{\theta_l=\theta_0}=0
    \end{cases}.
    \end{align}
     
    By using  (\ref{eqbeta}), it can be verified that
    \be
    \label{first_der}
    \left.\frac{\partial 	 \beta_{\Phi}(\theta_l,\thetavec_{n})}{\partial \theta_l}\right|_{\theta_l = \theta_{0}}  =\left. \frac{\partial}{\partial \theta_l} \left(   \int_{\Omega_\xvec} \Phi(x) f(\xvec; \theta_l, \thetavec_{n})d\xvec \right)\right|_{\theta_l = \theta_{0}} .
    \ee
     Under the assumption that the test, $\Phi(x)$, is independent of the parameter $\theta_l$, 
     the integration and  derivatives in \eqref{first_der} can be reordered 
     to obtain
        \begin{align}\label{eqfirst}
     \left.\frac{\partial 	 \beta_{\Phi}(\theta_l,\thetavec_{n})}{\partial \theta_l}\right|_{\theta_l = \theta_{0}}  =
     \int_{\Omega_\xvec} \Phi(x)   \left( \frac{\partial f(\xvec; \theta_l, \thetavec_{n}) }{\partial \theta_l}  |_{\theta_l = \theta_{0}} \right) \ud \xvec.
    \end{align} 
   Similarly,
       \beqna\label{eqsecond}
    \left.\frac{\partial^{2} 	 \beta_{\Phi}(\theta_l,\thetavec_{n})}{\partial \theta_l^{2}}\right|_{\theta_l = \theta_{0}}  
    =
     \int_{\Omega_\xvec} \Phi(x)   \left( \frac{\partial^{2} f(\xvec; \theta_l, \thetavec_{n}) }{\partial \theta_l^{2}}  |_{\theta_l = \theta_{0}} \right) \ud \xvec.
    \eeqna
   Therefore, by  substituting \eqref{eqfirst} and \eqref{eqsecond} in (\ref{OP1_2}), the integral form of (\ref{OP1_2}) is
    \begin{align} \label{OP}
    &\max_{\Phi(\xvec) } \int_{\Omega_\xvec} \Phi(\xvec) \left( \frac{\partial^2 f(\xvec; \theta_l, \thetavec_{n})}{\partial \theta_l^2}|_{\theta_l = \theta_{0}} \right) \ud \xvec  \\\nonumber
    s.t.~& \begin{cases}
    \int_{\Omega_\xvec} \Phi(\xvec) f(\xvec; \theta_{0}, \thetavec_{n}) \ud \xvec = \alpha \\\nonumber
    \int_{\Omega_\xvec} \Phi(\xvec) \left( \frac{\partial f(\xvec; \theta_l, \thetavecsmall_{n})}{\partial \theta_l}|_{\theta_l = \theta_{0}} \right) \ud\xvec = 0
    \end{cases}.
    \end{align}
   
  By using the  auxiliary lemma of the 
    Generalized Neyman-Pearson lemma (see  p. 77 in  \cite{Lehmann})
with  $m=2$, $f_{1} = f(\xvec; \theta_0, \thetavec_{n}) $, $f_{2} =\frac{\partial f(\xvec; \theta_l, \thetavecsmall_{n})}{\partial \theta_l}|_{\theta_l = \theta_{0}}$,$f_{3} =\frac{\partial^2 f(\xvec; \theta_l, \thetavecsmall_{n})}{\partial \theta_l^2}|_{\theta_l = \theta_{0}}$ , $c_1=\alpha,~c_2=0$, the LMPU test which solved (\ref{OP})  rejects the null hypothesis when
    \begin{align}\label{soul}
        \left.\frac{\partial^2 f(\xvec; \theta_l, \thetavec_{n})}{\partial \theta_l^2}\right|_{\theta_l = \theta_{0}}> k_{1} \left.\frac{\partial f(\xvec; \theta_l, \thetavec_{n})}{\partial \theta_l}\right|_{\theta_l = \theta_{0}} + k_2 f(\xvec; \theta_0, \thetavec_{n}).
    \end{align}
  It can be verified that
    \begin{align}\label{eqf1}
        \left.\frac{\partial f(\xvec; \theta_l, \thetavec_{n})}{\partial \theta_l}\right|_{\theta_l = \theta_{0}} = \left.\frac{\partial \log  f(\xvec; \theta_l, \thetavec_{n})}{\partial \theta_l}\right|_{\theta_l = \theta_0} f(\xvec; \theta_l, \thetavec_{n})|_{\theta_l = \theta_0}
    \end{align}
    and
    \begin{align}\label{eqf2}
        \left.\frac{\partial^2 f(\xvec; \theta_l, \thetavec_{n})}{\partial \theta_l^2}\right|_{\theta_l = \theta_{0}}  = \left(  \frac{\partial^2 \log f(\xvec; \theta_l, \thetavec_{n})}{\partial \theta_{l}^2}  \right)|_{\theta_l = \theta_0}  f(\xvec; \theta_l, \thetavec_{n})|_{\theta_l = \theta_0} +  \frac{\partial f(\xvec; \theta_l, \thetavec_{n})}{\partial \theta_l}|_{\theta_l = \theta_{0}} \frac{\partial \log f(\xvec; \theta_l, \thetavec_{n})}{\partial \theta_l}|_{\theta_l = \theta_0}.
    \end{align}
   By substituting \eqref{eqf1} into \eqref{eqf2}, one obtains
    \begin{align}\label{eqf3}
       \frac{\partial^2 f(\xvec; \theta_l, \thetavecsmall_{n})}{\partial \theta_l^2}|_{\theta_l = \theta_{0}} = \left(  \frac{\partial^2 \log f(\xvec; \theta_l, \thetavec_{n})}{\partial \theta_{l}^2}  \right)|_{\theta_l = \theta_0} f(\xvec; \theta_l, \thetavec_{n})|_{\theta_l = \theta_0} + \left( \frac{\partial \log f(\xvec; \theta_l, \thetavec_{n})}{\partial \theta_l}|_{\theta_l = \theta_0} \right)^{2} f(\xvec; \theta_l, \thetavec_{n})|_{\theta_l = \theta_0}.
    \end{align}
    Then, by substituting \eqref{eqf1} and \eqref{eqf3} in \eqref{soul}, we get that the LMPU test which solved (\ref{OP}) is the test in \eqref{eqlmpu}.
    
	\bibliographystyle{IEEEtran}
 
\end{document}